\documentclass[sigconf,10pt]{acmart}

\usepackage{hyperref}
\usepackage{cleveref}
\usepackage{subcaption}
\usepackage{balance}
\usepackage{url}
\usepackage{xspace}
\usepackage{amsmath}
\usepackage{tikz}
\usepackage{multirow}
\usepackage{bigints}

\usepackage{etoolbox}
\appto\UrlBreaks{\do\-}

\crefname{section}{\S\!}{\S\S}
\crefname{figure}{Figure}{Figures}
\crefname{table}{Table}{Tables}

\PassOptionsToPackage{hyphens}{url}

\newcommand{\sys}{ProLoc}

\newcommand{\device}{\ensuremath{d}}
\newcommand{\trustanchors}{\ensuremath{\mathit{TA}}}
\newcommand{\locc}{\ensuremath{\theta}}
\newcommand{\radius}{\ensuremath{r}}

\newcommand{\regionA}{\ensuremath{R}}
\newcommand{\region}{\regionA}
\newcommand{\regionB}{\ensuremath{S}}
\newcommand{\regionof}[2]{\region_{#1}^{#2}}

\newcommand{\kw}[1]{\ensuremath{\mathsf{#1}}}
\newcommand{\api}[1]{\mbox{\texttt{#1}}}

\title{
  \sys: Robust Location Proofs in Hindsight
}

\def\eg{{e.g., \xspace}}

\newif\iftr
\trtrue
\iftr
\usepackage{amsthm}
\fi

\newcommand{\btitle}[1]{{\iftr\medskip\else\vspace{0.3\baselineskip}\fi\noindent\textbf{#1}}}

\begin{document}
\author{Roberta De Viti}
\authornote{These authors contributed equally to this work and are listed alphabetically.}
\author{Pierfrancesco Ingo}
\authornotemark[1]
\author{Isaac Sheff}
\author{Peter Druschel}
\author{Deepak Garg}
\affiliation{
	\institution{Max Planck Institute for Software Systems (MPI-SWS)\\Saarland Informatics Campus}
	\streetaddress{}
	\city{Saarbrücken}
	\state{}
	\country{Germany}
	\postcode{}
}

\begin{abstract}
	Many online services rely on self-reported locations of user devices
like smartphones. To mitigate harm from falsified self-reported
locations, the literature has proposed \emph{location proof services}
(LPSs), which provide proof of a device's location by corroborating
its self-reported location using short-range radio contacts with
either trusted infrastructure or nearby devices that also report their
locations.

This paper presents {\sys}, a new LPS that extends prior work in two
ways. First, {\sys} relaxes prior work's proofs that a device
was \emph{at} a given location to proofs that a device
was \emph{within distance $d$} of a given location. We argue that
these weaker proofs, which we call \emph{region proofs}, are important
because (i) region proofs can be constructed with few requirements on
device reporting behavior as opposed to precise location proofs which
cannot, and (ii) a quantitative bound on a device's distance from a
known epicenter is useful for many applications. For example, in the
context of citizen reporting near an unexpected event (earthquake,
fire, violent protest, etc.), knowing the verified distances of the
reporting devices from the event's epicenter would be valuable for
both ranking the reports by relevance and flagging fake reports.

Second, {\sys} includes a novel mechanism to prevent collusion attacks
where a set of attacker-controlled devices corroborate each others'
false locations. Ours is the first mechanism that does not need
additional infrastructure to handle attacks with made-up (virtual)
devices, which an attacker can create in any number at any location
without any cost. For this, we rely on a variant of the TrustRank
algorithm applied to the self-reported trajectories and encounters of
devices. Our specific goal is to prevent \emph{retroactive} attacks
where the adversary cannot predict ahead of time which fake location
it will want to report, which is the case for the reporting of
unexpected events.

\end{abstract}

\maketitle

\section{Introduction}
\label{sec:intro}

Nowadays anyone with a smartphone can capture, document, and
disseminate news. This collective power to share real-time,
unfiltered, accounts of events is the essence of \textit{citizen
  journalism}. By democratizing the reporting of information, citizen
journalism represents a seismic shift from traditional media and has
been reshaping the landscape of news reporting, with social media
platforms like Instagram, Facebook, and Reddit serving as outlets for
news dissemination.

Citizen journalism plays a crucial role in offering pervasive,
on-the-ground perspectives, particularly during crises or unexpected
events. For instance, the Iranian presidential election protests in
2009 and the Egyptian uprising in 2011 showcased the profound impact
of citizen reporting, and its capability to bypass government
censorship and mobilize collective
action~\cite{britannica}. Similarly, the 2005 London bombings were
extensively covered by citizens’ on-site photos and videos, which
traditional news outlets later incorporated, recognizing the value of
the information that citizens provided~\cite{yale}.

However, citizen journalism is not subject to the same journalistic
scrutiny as news in mainstream news media. This means citizen
reports may be of varying quality, reflect bias, and even present an
avenue for deliberate propaganda and misinformation.  Therefore, the
ability to verify the quality and legitimacy of citizen journalistic
reports is crucial. However, doing so presents a complex,
multi-layered challenge. We take an important step towards addressing
this challenge: validating the location of \textit{geo-tagged}
reports.

Verifying the geolocation of a citizen at the time of an event is
relevant for both the quality and veracity of the citizen's report
about the event. First, a verified geolocation speaks to the quality
of the report, because a citizen who was near the event is more likely
to have observed the event first-hand.  Second, citizen journalism is
most relevant for events whose time and location is not predictable
(e.g., natural disasters, accidents, attacks) because mainstream media
and journalists are not likely to be present initially. For such
reports, a verified space-time location speaks to the authenticity of
the report and any associated digitally captured material (audio,
video, still images). This is because an adversary intent on
submitting biased or false reports about such an event would have to
have a device physically near the (unpredictable) time and location of
the event, which is a challenge even for powerful adversaries like
troll farms.

\btitle{Requirements and prior work.}  Prior work has proposed several
\emph{location proof services}\footnote{Location proofs are not proofs
in the mathematical sense; we adopt the term because it is used in the
literature.} or LPSs that provide independent evidence of a device's
(self-reported)
geolocation~\cite{saroiu2009enabling,veriplace,Hasan2011WhereHY,crepuscolo,applaus,link,stamp,globalattestation}.
However, citizen journalism and similar scenarios place specific
requirements on a LPS. First, the system should be able to generate a
location proof for an arbitrary time and location a user
visited. Second, the system must be robust to collusion attacks that
seek to generate false location proofs using many fictitious devices.
Existing LPSs are not a good match for these requirements. Some systems
depend on radio contact with trusted WiFi base stations at known
locations, which is unlikely to be available at the location of an
unpredictable
event~\cite{saroiu2009enabling,veriplace,Hasan2011WhereHY}. Other
systems instead depend on short-range radio contacts with nearby
devices that can corroborate each others' location
reports~\cite{link,stamp,applaus}. However, they require radio contact
with nearby devices at the precise time and location in question. If
such radio contacts exits, then a proof for the specific space-time
location is issued; otherwise, no proof can be issued at all. Lastly,
while some existing LPSs include partial defenses against collusion
attacks, none are robust against large-scale collusion by fictitious
devices.

\btitle{\sys.}  To this end, we present a new LPS, {\em \sys}. For a
given device $d$, instant in time $t$, and a required number of {\em
  witnesses} $N$, {\sys} can determine retroactively a {\em feasible
  region $S$} within which $d$ could have been present at time $t$,
given the evidence collected from $N$ witnesses.  Witnesses are unique
devices that can independently corroborate $d$'s presence around
$t$. {\sys} determines the feasible region by correlating, during a
period around $t$, (1) short-range radio encounters between $d$ and
nearby witness devices; (2) the geolocations recorded by $d$ and the
witnesses around $t$; and (3) uses the transportation network map in
and around the feasible region to see how far a device could have
traveled between radio contacts and geolocation reports.

By weakening (discrete) location proofs to regions proofs, {\sys} is
able to utilize short-range radio contacts and geolocations that were
recorded {\em around} but not necessarily {\em precisely at} time
$t$. Therefore, it is able to produce location proofs in many
scenarios where existing systems cannot. On the other hand, we will
show that region proofs are useful towards assessing the veracity of
citizen journalistic reports.

We note that media platforms that support citizen journalism typically
record geolocations periodically reported by users already.  {\sys}
additionally relies on Bluetooth Low Energy (BLE) advert exchanges,
which smartphones can emit and record with high energy efficiency.

LPSs that do not rely on trusted infrastructure are susceptible to
collusion attack. In particular, a device can fake radio contacts with
colluding devices to be able to obtain a fake location proof.  {\sys}
focuses on \textit{retroactive attacks}, where the adversary cannot
predict upfront for what time and location it wishes to generate a
false proof.  We consider this assumption consistent with scenarios
like citizen journalism, which is most relevant for unexpected events
–-- otherwise, traditional news websites would have their own
journalists present at the event.

An adversary may create a large number of \textit{fictitious devices}
that can corroborate false location reports.  {\sys} enables defences
against a large class of these collusion attacks involving fictitious
devices.  {\sys}'s defence relies on trust computations based on
devices’ connectivity graph. To the best of our knowledge, {\sys} is
the first infrastructure-less LPS robust to large-scale retroactive
attacks with fictitious devices.

We believe that {\sys} provides an important step towards enhancing
the reliability of citizen journalism reportings.  ProLoc offers: (i)
a practical method to provide region proofs in hindsight, based on
corroborating evidence recorded around the time of an event; and (ii)
the first infra\-structure-less defence against a large class of
easy-to-mount collusion attacks with many fictitious
devices, and a limited number of adversarial physical devices.

This paper's contribution includes (i) {\em \sys}, a system that
validates location reports retroactively, and produces region proofs
based on the device's radio contacts around the time on question; (ii)
an analysis of powerful collusion attacks and the first defense that
can tolerate large-scale collusion, including multiplicity and
stalking, without requiring trusted infrastructure; and, (iii) an
experimental evaluation based on a simulated dataset and on the DTU
dataset of locations and BLE contacts collected from 850 devices over
a period of 3 years.  In the remainder of this paper, we provide an
overview of \sys's design in \cref{sec:design-prelim}. Then, we
discuss \sys's API in~\cref{sec:service}, which describes how we
consolidate user trajectories to generate location proofs, and we
present our defense against large-scale collusion attacks
in~\cref{sec:dev-trust}.  We report an experimental evaluation of
\sys\ in~\cref{sec:data} and ~\cref{sec:eval}, and discuss related
work in~\cref{sec:related}.

\section{Design preliminaries}
\label{sec:design-prelim}

\label{sec:background}

\btitle{Location proof vs. location provider.}  A {\em location proof
  service} (LPS) like \sys\ is orthogonal to a {\em location
  provider service}; the two serve different purposes and may provide
different results depending on the scenario. For instance, the
location provider of a smartphone carried by a lone hiker in a
national park may provide location with a precision of a few meters,
which is great for navigation and even to provide coordinates in an
emergency call. However, if the hiker provides this location to the
third party, then, from the third party's perspective, the only root
of trust for the location is the hiker's phone. A location proof
service, on the other hand, could provide at best place the hiker
somewhere in a large feasible region based on the most recent contacts
the phone had at the parking lot and the distance the hiker could have
traveled in the interim. Such a proof could still be useful if the
hiker wishes to subsequently prove to a third party that they were in
or near the park. On the other hand, the smartphone of a user in a
busy downtown hotel not only obtains a precise location, but may be
able to obtain a location proof with a tight feasible region placing
him within tens of meters of a hotel, because of its numerous BLE
contacts with a large and diverse set of devices. As a result, the
user can prove, for instance, that they actually were close to the
hotel when they claim to have witnessed a fire in the hotel.

\btitle{State-of-the-art.}  As detailed~\cref{sec:related},
prior work has shown that devices can obtain location proofs from WiFi
APs or from nearby devices via short-range
radio~\cite{saroiu2009enabling,veriplace,Hasan2011WhereHY,crepuscolo,applaus,link,stamp,globalattestation}.
However, existing systems can attest locations only when a location
proof was obtained {\em a priori} from nearby APs or devices at the
time, or in the case of~\cite{globalattestation}, when a
\emph{trustworthy} device was present at the exact time. Since proof
generation is an expensive operation, it is not feasible to
continuously generate proofs for all locations visited to anticipate a
possible future need. Hence, existing systems cannot be used to
  generate location proofs \emph{a posteriori} (without advance
  planning), as would be needed in the context of citizen
  journalism. In contrast, we designed {\sys} to generate proofs a
  posteriori.

To mitigate collusion attacks, some prior work has used plausibility
and consistency checks on device trajectories, radio contacts, witness
selection, and voting
record~\cite{applaus,link,stamp,globalattestation}.
However, no existing work is robust against collusion attacks
involving large numbers of fictitious devices that corroborate each
others' statements. {\sys} includes a defense that works against
  retroactive collusion attacks even with large numbers of fictitious
  devices.

\btitle{\sys\ model.}
Every device participating in {\sys} locally stores and periodically
uploads its {\em device history}, which is a time series of
geographical {\em locations} and received {\em near-range radio
  transmissions} from peers.
Locations originate from the location service of the device's platform
(e.g., iOS or Android), which typically relies on a combination of
GPS, cellular, and WiFi-based methods.  Locations are 
recorded on a regular basis (e.g., every 5 minutes while a device is
on).

Near-range transmissions occur over short-range radio like Bluetooth
Low Energy (BLE).  As explained in prior
work~\cite{SDDR,crepuscolo,enclosure}, near-range transmissions
typically take the form of an ephe\-me\-ral id, a high
entropy value unique to the transmitting device and epoch.  A device
history is a time series of location reports (the device's trajectory)
interspersed with ephemeral ids received from nearby devices.

An {\em encounter} is a period of co-location between two devices,
called the encounter {\em peers}. Individual received ephe\-me\-ral ids
do not directly constitute encounters.
An encounter is indicated
only when both peers receive each others' transmissions regularly for
the duration of the encounter.  Detecting encounters requires
correlating the histories of potential encounter peers by the
backend.  Encounters tie histories of different devices to each other,
as illustrated in~\cref{fig:encountergraph}.

\begin{figure}[t]
	\centering
  \caption{User device trajectories (colored lines) containing
           location reports, connected by encounters.}
	\label{fig:encountergraph}
	\includegraphics[width=0.32\textwidth]{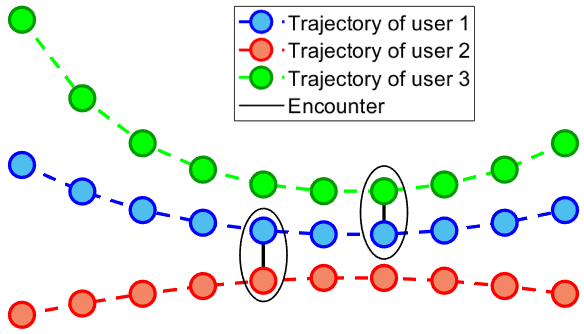}
  \vspace{-5mm} %
\end{figure}

\btitle{Assumptions.}
{\sys} integrates into existing mobile services that already
  collect location histories of devices. {\sys} requires devices to
  additionally collect and upload pseudonymous transmissions received
  over near-range radio from peer devices, and requires the service
  backend to implement the location proof algorithm
  (\cref{sec:service}) and the defense against collusion attacks
  (\cref{sec:attacks-sys}).  
  We note that these services that possess user location data
  can already infer periods of co-location; the addition of
  BLE data further corroborates these interactions without significantly
  expanding the platforms' existing knowledge on their users.

Devices are expected to upload their history within minutes of when
events happened, effectively asking devices to commit to their
trajectory and encounters immediately. Early commitment denies an
adversary scope for certain retrospective attacks, where it rewrites
the history of devices it controls to suit emerging objectives. The
requirement is reasonable, because most smartphones are connected
continuously, the bandwidth requirements are modest, and no
  user action is required. Late uploads of transmission received
  from peers are accepted, but an encounter can be used in a proof
only if at least one of the peers uploaded the encounter in a timely
fashion.

Our defense against collusion attacks relies on the detection of suspicious devices based on their connectivity to other devices. This fundamentally
  requires knowledge of some
  trusted devices as trust anchors. As we will show in
  \S\ref{sec:eval}, a very small number of devices is sufficient
  in \sys, and the set of devices can be different for each
  verifier. A verifier can choose their set freely from devices they
  trust, e.g., devices owned/operated by their own organization.

\btitle{Power considerations.}  Location traces can be collected
in a power-efficient manner using standard platform services. Sending
and capturing BT advertisements in a power-efficient (and
privacy-preserving) manner has been demonstrated in prior
work~\cite{SDDR,enclosure} and is now widely used in COVID-19 contact
tracing apps.

\subsection{Key Insights}
The key insight behind {\sys} is that continuous collection,
  aggregation, and analysis of device locations and encounters can
  achieve properties not available in existing systems:

\btitle{1) \sys\ can provide a feasible region for
    each instant along a device's trajectory {\em a posteriori}}. To
  compute this region at instances when a device did not encounter
  other devices, \sys\ extrapolates from the closest previous and next
  confirmed locations and the maximal distance the device may have
  traveled in the interim. To compute this region, \sys\ intersects
  \emph{isochrones} obtained from the OpenStreetMap~\cite{openstreetmap}
  service. An isochrone bounds the set of locations reachable from a
  starting point within some period, in any direction, given
  the topology of the road and transportation network and the
    feasible speed of travel along each route.

\btitle{2) \sys\ can mitigate retroactive collusion attacks involving
  fictitious devices.}  Because {\sys} requires devices to upload and
thereby commit to their locations and encounters in a timely fashion,
an attacker faces the difficulty of having to produce plausible device
trajectories and encounters that afford its devices enough trust to
appear legitimate, and are suitably positioned for an attack.
Moreover, we sketch how the defense can be extended to cover {\em
  premeditated} collusion attacks as well.

\section{{\sys} Service}
\label{sec:service}
\label{sec:loc-conf}

{\sys}'s backend allows any participating device, called the
\emph{prover}, to prove a posteriori to a third-party, the
\emph{verifier}, that it was within a given region $\regionB$ at time
$t$. The backend's only API call,
$\api{prove\_loc}(\device, \regionB, t, \trustanchors, N)$, 
returns true when {\sys} is able to prove that device $\device$ was
definitely within region $\regionB$ (specified as a bounding polygon
on a map) at time $t$ and false otherwise.
The parameters $\regionB$, $t$, $\trustanchors$ and $N$ are
  determined by the \emph{verifier}. $\trustanchors$ and $N$ are
  needed for {\sys}'s defense against collusion attacks
  (\cref{sec:attacks-sys}). $\trustanchors$ is a (small) set of
  \emph{trust anchors} -- devices that the verifier knows to be
  honest. As explained in \S\ref{sec:dev-trust}, {\sys} uses
  $\trustanchors$ to seed a TrustRank algorithm which identifies
  suspicious devices. $N$ is the number of encounter peers which
    \emph{independently} corroborate $d$'s presence in $S$ at time
    $t$. A higher $N$ increases {\sys}'s robustness to collusion
  attacks, but may also increase the sizes of regions $\regionB$ for
  which location proofs can be successfully found.

To answer the call $\api{prove\_loc}(\device, \regionB, t,
  \trustanchors, N)$,
  {\sys}'s backend picks $M > N$ non-suspicious device peers which
  reported encounters with $\device$ close to time $t$. These peers,
  called \emph{potential witnesses}, are denoted $w_1, ..., w_M$
  here. For each potential witness $w_i$, the backend computes a
  region $\regionA_i$ in which $\device$'s presence at time $t$ can be
  established using information from $w_i$ alone: $w_i$'s location
  history and its encounters with $d$. If $\regionA_i$ is contained in
  $\regionB$, then the potential witness $w_i$ is a \emph{valid
  witness} for the proof. The call $\api{prove\_loc}(\device,
  \regionB, t, \trustanchors, N)$ returns true (successfully) if and
  only if the backend finds at least $N$ valid witnesses, i.e., if the
  \emph{union} of some $N$ regions out of the $M$ regions $\regionA_1,
  \ldots, \regionA_M$ is contained in $\regionB$. The $N$ valid
  witnesses form the proof's \emph{quorum} and this union is called
  $d$'s \emph{feasible region} at time $t$.

The region $\regionB$ may be a circle specified as a center and
radius, $(\locc, \radius)$. The API call $\api{prove\_loc}(\device,
(\locc,\radius), t, \trustanchors, N)$ then asks, ``Can $N$ devices'
histories independently confirm that $d$ was present within distance
$\radius$ of the point $\locc$ at time $t$?'' In this case, we also
call $\radius$ the \emph{precision}. Note that smaller $r$ correspond
to better precision.  

For privacy reasons, a device may ask only for its own location
proofs, i.e., a device may act as prover only for itself. However,
{\sys} cryptographically signs the result (true or false) of the call
together with the parameters to allow the device to prove its location
to the verifier.

Next, we describe {\sys}'s proof algorithm in detail.

\btitle{Notation.}  We use the letter $\locc$ and its variants like
$\locc'$, $\locc_w$, etc.\ to denote points on a map. Similarly, we
use $t$ to denote time points, $\radius$ to denote lengths, and
$\region, \regionB$ for regions.

\btitle{Mapping service.}  \sys\ relies on a map service
like OpenStreetMap~\cite{openstreetmap}. This service provides a call
$\api{isochrone}(\region, T)$ which returns the entire region a device
could have reached in time $T$ starting somewhere in region $\region$
using any available means of transport at the fastest speeds possible.

\begin{figure}[t]
  \definecolor{comment}{rgb}{0,0.4,0}
  \definecolor{function}{rgb}{0,0,0.6}
  \definecolor{api-call}{rgb}{0.4,0,0}
  \definecolor{return}{rgb}{0.4,0,0}
  \definecolor{map}{rgb}{0.7,0.3,0.2}
  \definecolor{D}{rgb}{0.5,0,0.5}
  \vspace{-3mm} %
  \caption{Pseudocode of $\api{prove\_loc}$.
    $\api{{\color{D}\normalfont Dbase}}$ is {\sys}'s internal database, and
    $\api{{\color{map}\normalfont Map}}$ is the map service.
  }
  \label{fig:conf-region}
  \begin{tabular}{@{}l@{}}
    \hline
    API call {\color{api-call}$\api{prove\_loc}(\device, \regionB, t, \trustanchors, N)$} \\
    \hline
    {\color{comment}\# Compute region for $\device$ w.r.t.\ a single peer location}\\
    {\color{function}$\kw{function}$ $\api{r\_peer\_loc}(t_e, t_w, \locc_w)$:} \\
    \hspace{4mm} $\regionof{\device_w}{t_e} \leftarrow \api{{\color{map}Map}.isochrone}(\{\locc_w\}, |t_e - t_w|)$ \\
    \hspace{4mm} $\regionof{\device}{t_e} \leftarrow \api{{\color{map}Map}.grow\_region}(\regionof{\device_w}{t_e}, r_{BLE})$ \\
    \hspace{4mm} ${\color{return}\kw{return}}(\api{{\color{map}Map}.isochrone}(\regionof{\device}{t_e}, |t - t_e|))$ \\
    {\color{comment}\# Compute region for $\device$ w.r.t.\ a single peer encounter}\\
    {\color{function}$\kw{function}$ $\api{r\_encounter}(E_i, \device)$:} \\
    \hspace{4mm} $w \leftarrow E.\api{other\_peer}(\device)$ \\
    \hspace{4mm} $t_e \leftarrow E.\api{timeof}()$ \\
    \hspace{4mm} $(\locc_w, t_w) \leftarrow \api{{\color{D}Dbase}.prev\_location\_report}(\device_w, t_e)$ \\
    \hspace{4mm} $(\locc_w', t_w') \leftarrow \api{{\color{D}Dbase}.next\_location\_report}(\device_w, t_e)$ \\
    \hspace{4mm} $\regionA \leftarrow {\color{function}\api{r\_peer\_loc}(t_e, t_w, \locc_w)}$ \\
    \hspace{4mm} $\regionA' \leftarrow {\color{function}\api{r\_peer\_loc}(t_e, t_w', \locc_w')}$ \\
    \hspace{4mm} $\regionA_i \leftarrow \regionA \mathrel{\cap} \regionA'$ \\
    \hspace{4mm} ${\color{return}\kw{return}}(\regionA_i, w)$\\
    {\color{comment}\# Top level code }\\
    $[E_1, \ldots, E_M] \leftarrow \api{{\color{D}Dbase}.encs\_select}(\device, t, \trustanchors)$\\
    $Q \leftarrow \emptyset$ {\color{comment}~~~~~~~~\# Quorum};  \hspace{1mm}
    $\region_F \leftarrow \emptyset$ {\color{comment}~~~~~~~~\# Feasible region}\\
    $\kw{for}~ i ~\kw{in}~ 1 \ldots M$:\\
    \hspace{4mm} $(\regionA_i, w_i) \leftarrow \api{r\_encounter}(E_i, \device)$\\
    \hspace{4mm} $\kw{if}~\regionA_i \subseteq \regionB~\kw{then}~ \{$
     $Q \leftarrow Q \mathrel{\cup} \{w_i\};$
    \hspace{1mm} $\region_F \leftarrow \region_F \mathrel{\cup}  \regionA_i\ \}$\\
    $\kw{if}~|W| \geq N ~\kw{then}~ \kw{return}(\kw{true}) ~\kw{else}~ \kw{return}(\kw{false})$\\
    \hline
  \end{tabular}
  \vspace{-5mm}
\end{figure}

\btitle{{\sys}'s algorithm.}
To answer the call $\api{prove\_loc}\linebreak[6](\device, \regionB,
t, \trustanchors, N)$, {\sys} first uses the \emph{encounter selection
algorithm} described later to select encounters $E_1, \ldots, E_M$
that $\device$ had with distinct peers in the vicinity of time
$t$. Let $w_1, \ldots, w_M$ be the respective encounter peers (these
are the potential witnesses). For each $E_i$, $1 \leq i \leq M$,
{\sys} uses the function $\api{r\_peer}(E_i, \device)$ described below
to construct a region $\regionA_i$ in which $d$'s presence at time $t$
can be established from information provided by $w_i$ alone. Next, for
each $i$, {\sys} checks whether $\regionA_i$ is contained in
$\regionB$ or not. If it is contained, then {\sys} has found a valid
witness. The algorithm stops with success if $N$ valid witnesses are
found, else it ends in failure.
\cref{fig:conf-region} summarizes the algorithm.

\btitle{Function $\api{r\_peer}(E_i, \device)$.} Let $t_{e}$ be the time at
which encounter $E_i$ happened and let $w$ be the encounter peer (the
potential witness). Let $\locc_w$ and $\locc_w'$ be the self-reported
locations of $w$ preceding and following $t_{e}$, and let these
locations be reported at times $t_w$ and $t_w'$, respectively.

We start from one of $\device_w$'s location reports, say, $\locc_w$ at
time $t_w$. We determine a possible region
$\regionof{\device_w}{t_{e}}$ for $\device_w$ (not $\device$) at the
time of the encounter, $t_{e}$. Since $w$ was at $\locc_w$ at time
$t_w$, this region is simply $\regionof{d_w}{t_{e}} =
\api{isochrone}(\{\locc_A\}, |t_{e} - t_w|)$.

Next, we compute a possible region $\regionof{\device}{t_{e}}$ for
$\device$ (not $\device_w$) at time $t_{e}$. Since $\device$
encountered $\device_w$ at this time, $\device$ must have been within
BLE range, say $r_{BLE}$ meters, of $\device_w$ at time $t_{e}$. So,
$\regionof{\device}{t_{e}}$ is obtained by expanding the region
$\regionof{\device_w}{t_{e}}$ on the map by $r_{BLE}$.
Then, we compute a possible region $\region$ for $\device$ at time $t$
as $\region = \api{isochrone}(\regionof{\device}{t_e}, |t - t_{e}|)$.

We repeat this process with $\device_w$'s second reported location
$\locc_w'$ at time $t_w'$. This yields a second possible region
$\region'$ for $\device$ at time $t$.
Since $\device$'s location at time $t$ is constrained by both
$\region$ and $\region'$, the required region $\regionA_i$ is the
\emph{intersection} $\regionA_i = \region \mathrel{\cap} \region'$.

\begin{figure}[t]
  \centering
  \vspace{-3mm} %
  \caption{Feasible region (shaded) w.r.t.\ one witness
    $w$. }
  \label{fig:feasible-region}
  \includegraphics[width=0.7\linewidth]{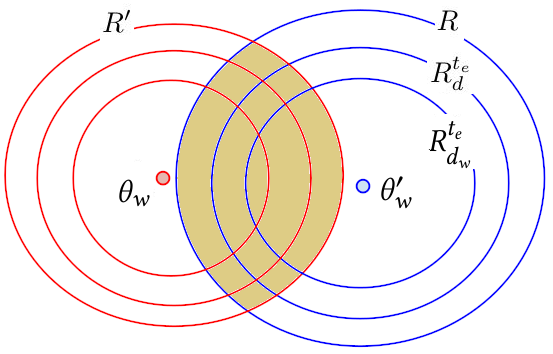}
  \vspace{-6mm} 
\end{figure}

\Cref{fig:feasible-region} illustrates this computation
diagramatically.

\btitle{Encounter selection algorithm.}  The encounter selection
algorithm, denoted $\api{encs\_select}(\device, t,
\trustanchors)$, selects encounters of $\device$ with distinct
non-suspicious devices in the vicinity of time $t$ and location
$\locc$. {\sys} first runs the TrustRank algorithm described in
\S\ref{sec:dev-trust} seeded with the trust anchors {\trustanchors} to
mark suspicious-looking devices. ({\sys} caches results of
TrustRank. All subsequent calls with the same {\trustanchors} and, in
particular, all subsequent calls from the same verifier reuse this
result without recomputing TrustRank.)

Next, we find all encounters of $\device$ in a large time window $[t -
  a, t + a]$ centered at $t$ ($a$ is fixed to \eg 5 mins). We remove
encounters with suspicious-looking peers and sort the remaining
encounters \emph{ascending} by the following function $f$. The sorted
list of encounters is the output of $\api{encs\_select}(\device, t,
\trustanchors)$.
\vspace{-2mm}
\[
f=|t_e-t|+\min(|t_e-t_w|,|t_e-t_w'|)+\frac{\min(|\locc-\locc_w|,|\locc-\locc_w'|)}{v}
\]
Here, $\locc$ is the self-reported location of the device $\device$ at
time $t$, $t_e$ is the time of the encounter, $\locc_w$ and $\locc_w'$
are the reported locations of the encounter peer preceding and
following the encounter, $t_w$ and $t_w'$ are the respective
timestamps of these location reports, and $v$ is the average travel
speed in the vicinity of $\locc$ according to the map service. The
function $f$ embodies the fact that an encounter constrains $\device$
more if: (i) the time between $t$ and the encounter is less (term
$|t_e-t|$), (ii) the encounter peer has reported a location closer in
time to the encounter (term $\min(|t_e - t_w|, |t_e - t_w'|)$), and
(iii) the reported peer location is close to $\device$'s actual
location $\locc$ (term $\min(|\locc-\locc_w|, |\locc-\locc_w'|)$).

\btitle{Precision.}  The smallest region $S$ that yields a successful
proof depends on four factors: (1) $N$: Higher values of $N$ provide
more collusion resistance but prevent a successful proof for small
$S$, (2) Local encounter density: If the proving device has many
encounters in the spatio-temporal vicinity of the proof location, then
the precision is likely to be better, (3) $r_{BLE}$: In places where
BLE encounters are limited to small distances due to obstacles and
noise, precision will be worse, and (4) Location reporting
frequency: Precision is better when the proving device's peers report
locations frequently.

In practice, $r_{BLE}$ is determined empirically by measuring the
effective BLE range in different conditions and choosing an upper
bound. Devices should be programmed to report locations at a fixed
interval, also determined empirically.

\btitle{Peer privacy.} The success of a location proof reveals one bit
of information about the witnesses' previous and next locations to the
prover and the verifier. To prevent the leakage of witness
trajectories using repeated proof calls, {\sys} rate-limits calls to
its API from every prover.

\section{Defending Against Retroactive Collusion Attacks}
\label{sec:attacks-sys}
\label{sec:dev-trust}

An adversary's goal in attacking a LPS is to prove that a device it
controls was at a target space-time location where it was actually not
present, e.g., to fraudulently prove that they eye-witnessed an
event. An attack is called \emph{premeditated} if the adversary knows
ahead of time where it will need a location proof later. To mount a
premeditated attack, the adversary only needs a few real-looking
colluding devices that -- correctly or falsely -- report presence at
the target location and encounters with the target device around that
location. Defending such attacks fundamentally requires additional
infrastructure, as has been proposed in prior
work (see \cref{sec:related}). However, premeditated attacks are less relevant
for use cases like citizen journalism, where the focus is often on
unexpected events whose location and time the adversary cannot predict
in advance.

Consequently, we focus on defending against the complementary class of
\emph{retroactive} attacks where the adversary cannot predict the
space-time location of interest in advance. Mounting such an attack
successfully requires continuous presence of adversarial devices at
all locations of potential interest. In turn, this requires the
adversary to control or simulate a very large number of devices.

\btitle{Adversarial devices.}  To mount a retroactive attack, the
adversary may deploy three kinds of (adversarial) devices.
(i) \textit{Corrupt} devices are real devices that the adversary
controls at least partially. For a subset, it may even control
physical whereabouts.
(ii) \textit{Sybil} devices, or Sybils, are virtual devices hosted on
corrupt devices (e.g., rotating IDs on a phone).
(iii) \emph{Fictitious} devices, have no correspondence to an actual
device. They are make pretend devices that report like real devices.
The term \emph{virtual} device includes both Sybil and fictitious
devices. (In contrast, we use the term \textit{honest} device for any device
that is not controlled by the adversary.)

All types of adversarial devices can report false location trajectories and
false encounters with other adversarial devices to enable
(retroactive) attacks. Additionally, corrupt and Sybil devices can
form real encounters with honest devices, but only along their real
physical trajectories (the honest peer will not confirm a false
encounter). Fictitious devices cannot form any encounters
with honest devices.

\btitle{Assumptions.}
Corrupt devices cost the adversary effort and money, so we assume that
their number is limited. Virtual devices cost nothing, so we do not
make any assumption on their number. Consequently, a successful
defense against retroactive attacks must identify most virtual
adversarial devices. The defense we present here does exactly
this.

\btitle{Key insight.}  \sys's defense relies on the insight that, even
though the adversary can fabricate any number of virtual devices with
arbitrary trajectories and encounters among themselves, any encounters
between this virtual world and the world of honest devices are limited
to corrupt and Sybil devices. Consider a graph of device encounters
(\cref{fig:encountergraph}).  In such a graph, the subgraphs
representing the virtual devices and honest devices form clusters that
are connected only by encounters involving corrupt and Sybil
devices. Since we assume that the adversary is limited to a small
number of corrupt devices and our defense limits the utility of Sybils
(as explained later), these encounters will form a small cut
between the adversarial and honest devices in the graph.

To arrive at a defense, we adapt the literature on random-walk-based
approaches to detect virtual
devices~\cite{DBLP:conf/dsn/JiaWG17,DBLP:conf/sp/AlvisiCELP13} and ban
them from quorums.  We use a particular random walk, namely, the
TrustRank algorithm~\cite{trustrank}.  We modify TrustRank to prevent
Sybils from artificially increasing the cut between the real and
virtual worlds. This allows us to recognize most fictitious devices
immediately (i.e., our defense has high recall). Furthermore, our
approach limits the number of Sybils an attacker can maintain without
detection.
In combination with a large enough $N$, this allows \sys\ to mitigate
retroactive attacks effectively.
While we did not design our defense for attacks where the adversary's
goal is to reduce the precision of location proofs of \textit{honest}
devices, our defense does, in fact, prevent such attacks.

\subsection{\sys's Defense}
\label{sec:trustrank}

\btitle{Recap of TrustRank.}  The original TrustRank
  algorithm~\cite{trustrank} takes as input a directed graph with
  positive edge weights, and an initial set of trusted nodes (which
  we denoted {\trustanchors} in \cref{sec:service}). It computes the
  answer to the following question: If (a) a random walk on the graph
  begins at a random trusted node, (b) at each step choosing to
  continue with a probability $\alpha$, by traversing an outgoing edge
  with probability proportional to the edge's weight, then, what is
  the probability that it will end at any given node? This ending
  probability assigned to each node is the node's TrustRank
  score. Naturally, nodes connected to trusted nodes only via
  low-weight cuts have lower scores compared to nodes connected to
  trusted nodes via high-weight cuts. So, TrustRank distinguishes
  these two kinds of nodes from each other.

\btitle{Our use of TrustRank.} We apply TrustRank to the
  \emph{encounter graph} whose nodes are devices and which has a
  directed edge from device $g$ to $d$ if $g$ received at least one
  BLE advert from $d$. (Note that edges are oriented opposite to the
  direction of adverts.) We seed TrustRank with the set of trusted
  devices {\trustanchors} provided by the verifier. We first explain
  how we use this setup for our defense assuming there are no Sybil
  devices (the only virtual devices are fictitious), and then explain
  how we handle Sybil devices.

In the absence of Sybils, the width of the cut between the
  honest and fictitious devices is low since it is limited by the
  number of corrupt devices, which is assumed to be low. Hence, in the
  absence of Sybils, we can run the TrustRank algorithm described
  above with the edge weight from $g$ to $d$ set to the number of
  $d$'s adverts received by $g$. This results in TrustRank scores with
  a bimodal distribution, with fictitious devices receiving
  significantly lower scores than real devices.

  To actually classify devices as real or fictitious, we automatically
  determine a threshold that separates the two modes, and classify all
  devices with scores above the threshold as non-suspicious (likely
  real) and those below as suspicious (likely virtual).  The threshold
  is determined by maximizing the sum of the percentages of true
  positives (real devices classified as non-suspicious) and true
  negatives (fictitious devices classified as suspicious) on a
  synthetic encounter graph that contains data from a set of devices
  which have been verified offline as being real (or not), and a set
  of synthetically generated attacker devices for different synthetic
  attack scenarios described in \cref{sec:evaltrustrank}. We show in
  \cref{sec:evaltrustrank} that honest and fictitious devices are so
  well-separated by TrustRank scores on a real dataset that simulating
  a few attacks is actually sufficient for determining this threshold,
  even when the attack(s) occurring in reality are unknown.

\btitle{Handling Sybils.} Sybils can defeat the defense
  described so far: A small number of corrupt devices can
  broadcast the identifiers of any number of Sybil devices in repeated
  succession, thus increasing the cut-width between the real and
  fictitious devices to any extent.  To defeat such attacks, we use a
  novel edge-weighting scheme that strongly penalizes large numbers of
  Sybils. We rely on the insight that Sybils on a single host appear
  to form concurrent encounters with nearby devices in such
  attacks. Our edge-weighting scheme heavily penalizes the
  contribution of concurrent encounters to edge-weights.

For each device, we divide the time series of received BLE adverts into
epochs of fixed duration (e.g., 8 min). We call two adverts received
by a device \emph{concurrent} if they are in the same epoch. Next, we compute
the weight of the edge from device $g$ to device $d$, written
$e(g,d)$, as follows: For each epoch in which $g$ received an advert
from $d$, we add to $e(g,d)$ 1 divided by the number of devices from
which $g$ received concurrent adverts raised to a power $L$, for a
fixed $L > 1$.%
\footnote{%
This edge-weighting scheme is a slight simplification of our actual
scheme, which uses continuous time and an integral instead of discrete
epochs and a sum over epochs. Continuous time gets rid of
epoch-boundary effects. The actual scheme is explained in \iftr
\cref{sec:edge-weighting} \else our (anonymized) TR~\cite{proloc-tr}\fi.}
\vspace{-1mm}
\[
e(g,d)
\hspace{-1mm}
\mathrel{\triangleq}
\hspace{-6mm}
\sum_{E \, \in\, \mbox{epochs}}
\hspace{-1mm}
\frac{
  \textrm{if }(g\textrm{ recd. advert from }d\textrm{ in }E)\textrm{ then }1\textrm{ else }0
}{
  \left(
  \textrm{\# devices from which }g\textrm{ recd.\ adverts in $E$}
  \right)^L
}
\]

To understand the effectiveness of the method, suppose that a corrupt
device with $m$ Sybils encounters a device $g$ in an epoch where $g$
receives adverts from $t$ other devices. Then, each Sybil ends up
getting an incoming edge weight $1/(t+m)^L$ from $g$ in this epoch and
the $m$ Sybils together get an incoming edge weight of $m/(t+m)^L$
from $g$ in this epoch. When $L > 1$, this function tends to $0$ very
quickly for large $m$.  Hence, by creating a large number of Sybils
$m$, the adversary lowers the incoming edge weights its Sybils receive
from $g$ in the epoch, relative to the incoming edge weights received
from $g$ by other devices in epochs where $g$ encountered honest
devices only. Since the total TrustRank score that the Sybils -- and,
hence, the entire virtual world -- gets is proportional to the Sybils'
total incoming edge weight relative to the total incoming edge weight
of all devices encountered by $g$ over time, it is in the adversary's
interest to pick small $m$, i.e., to limit the scale of its Sybil
attack.

We show experimentally in \cref{sec:evaltrustrank} that by setting
$L=3$, we can detect nearly all fictitious devices even in the presence
of strong adversaries that corrupt up to 0.5\% of real devices.

\btitle{End-to-end defense.} This leaves the adversary with only
corrupt and Sybil devices at its disposal. Corrupt devices are limited
in number by assumption. As for Sybils, the ideal strategy for an
attacker in the face of our edge-weighting scheme is to rotate them at
the granularity of epochs, in a way that they all get roughly equal
TrustRank scores. If the average TrustRank score of real devices
(including corrupt devices) is $s$ and the TrustRank
threshold determined above is $T$, then the adversary can sustain up
to $s/T$ Sybils per corrupt device with this strategy. In
\cref{sec:evaltrustrank}, we show that this ratio is around 10.6 for a
real dataset. With this, if the adversary can corrupt $c$ devices, it
ends up with $cs/T$ adversarial devices. \iftr\else(\fi In
  fact, we prove \iftr \else in our anonymized TR~\cite{proloc-tr} \fi that the total TrustRank
  across all adversarial devices is bounded by a constant factor of
  the sum of the TrustRanks corrupt devices can collect\iftr\else.)\fi
  \iftr:

\begin{theorem}
The sum of the TrustRank of attacker devices is at most proportional to the sum of the TrustRank of corrupt devices. 
\end{theorem}
\begin{proof}
  TrustRank of a device can be characterized as the probability that a random walker halts on the device, given:
  \begin{itemize}
    \item The walker begins on a (uniformly random) trusted device.
    \item At each step, the walker proceeds with some probability $\alpha$, and halts with probability $1-\alpha$.
    \item At each step, if the walker proceeds, it steps to a neighboring device with probability proportional to edge weight (defined above).
  \end{itemize}
  Since trusted devices are real, and the only attacker devices that encounter real devices are corrupt devices, a random walker that halts on an attacker device must have (at some point) stepped to a corrupt device.
  \\
  Therefore the probability that a random walker halts on an attacker device is, at most, the probability that a random walker steps to any corrupt device. 

  \paragraph{TrustRank is bounded by a constant factor of step probability.}
  A  random walker that steps to a device has at least a $1-\alpha$ probability of halting on a that device. 
  Therefore the TrustRank of a device is at least $1-\alpha$ (which is a constant) times the probability that a random walker steps to that device. 

  Therefore, the probability that a random walker steps to any corrupt device is at most proportional to the sum of the TrustRank of corrupt devices.
  This in turn implies that the sum of the TrustRank of attacker devices is at most proportional to the sum of the TrustRank of corrupt devices. 
\end{proof}
\fi

To limit how successfully an adversary can launch a retroactive 
attack, \sys\ additionally relies on the parameter $N$. With $cs/T$
colluding devices as described above, the adversary can hold a quorum
of $N$ adversarial devices at no more than $cs/(T\cdot N)$ locations
simultaneously. Just by way of example, an adversary that corrupts
$c=5$ devices, setting $N = 10$, with $s/T = 10.6$ as for our dataset,
the adversary can cover only four locations at a time. To be
successful in a retroactive attack, the adversary would have to
correctly narrow its guess of the locations that could be relevant for
an attack in the future to just a set of four.

\btitle{Possible extension to premeditated attacks.}
\label{sec:premed}  
Recall that in a premeditated attack, the adversary may move a
  quorum of adversarial devices to the target location on
  demand. Here, the small number of adversarial devices permitted by
  \sys's defense are sufficient for such an attack.  While beyond the
scope of this paper, the key to a defense against premeditated attacks
that does not require infrastructure is to prevent the adversary from
moving Sybils from the actual locations of the corrupt devices (which
it cannot easily choose) to the target location, while maintaining
their trust levels. One way to accomplish this is to diminish trust
whenever a device moves in space-time, another is to tie trust to
certain spatial zones of validity. A full design and evaluation of
such a defense remains future work.

\section{Datasets}
\label{sec:data}

\sys\ takes as input the device history of each participating
device. The device history is a time series of locations and BLE
receptions from nearby devices (\cref{sec:background}).
In this section, we introduce the datasets we use in the
experimental evaluation of~\cref{sec:eval}: a simulated dataset
and a real-world dataset.

\btitle{Simulated dataset.}
We use a mobility simulator from prior work \cite{pancastnature}
together with OpenStreetMap data \cite{openstreetmap} to synthesize
device locations and encounters in Co\-pen\-ha\-gen  over a 15-day period. 
This involves three main steps: (i) we use the simulator to create a list of
(hypothetical) users living in various Copenhagen
neighborhoods, in accordance with Meta's publicly
available data on population density, household, and age demographics
\cite{dataforgoodmeta}; (ii) we use OpenStreetMap data to identify places of
interest (POIs) on the map, labeled according to their location type
(e.g., supermarkets, schools, train stations); (iii) we use the
simulator and Meta's publicly available mobility models (which provide
the frequency of visits to each location type for different age
groups) to synthesize individuals' mobility traces as a series of
visits to POIs. A visit is a four-tuple consisting of a user device
ID, a location or POI, a visit start time, and a visit end time. A
user selects among POIs of the same location type with frequency
inversely proportional to the distance between the user's home and the
respective POIs---the so-called gravity
model~\cite{haynes2020gravity}.

\textit{Location dataset.} While at a POI, a user device records its
location with a given average frequency, which we vary in our
experiments from 1 to 5 minutes. Each record is a triple (device ID,
location, timestamp). Devices do not record locations at home and in
transit; this limited recording of locations is conservative as it
makes our evaluation results worse than they would be had user devices
continuously recorded locations. All the location records of a device
constitute its location history.

\textit{Encounters.} We cannot simulate individual BLE adverts
because successful receipt of a BLE transmission is highly dependent
on the physical environment, which is not covered in the mobility
models we use. Instead, we adopt a conservative approach: we assume
that if two devices remain within $r_{BLE} = 50$\,m of each other for
at least 5 minutes, then each of them will receive at least one BLE
advert from the other, thereby establishing a mutual encounter. Our
encounter dataset is the list of all such encounters represented as
triples (device ID$_1$, device ID$_2$, timestamp), where timestamp is
the midpoint of the time period during which the devices were in
proximity.

In practice, not everyone will adopt {\sys}, so we generate mobility
(location and BLE receipt) traces for only a chosen percentage of
randomly chosen users (the \emph{adoption rate}).

As an example, for an adoption rate of $20\%$ and an average location
reporting frequency of $3$ minutes, our dataset has 128,882 unique
participating devices (users), which together produce 72,905,340
location reports and 19,343,782 pairwise encounters. Over 15 days,
each device participates in 150 encounters on average (min: 0, median:
106, max: 3,000) and provides 565.67 location reports on average (min:
33, median: 543, max: 2,670).

\btitle{Real-world dataset.}
As mentioned above, we cannot simulate individual BLE adverts. Thus, to evaluate
TrustRank, which relies on individual adverts, we
use the SensibleDTU~\cite{dtudatasets} dataset of BLE adverts 
from Danmarks Tekniske Universitet
(DTU), Copenhagen.%
\footnote{Data collection, anonymi\-zation, and storage were
approved by the Danish Data Protection Agency. One of the authors had
access to the dataset under a collaboration agreement with DTU.}
The data was collected from March 2013 to August 2016 by DTU students
carrying LG Nexus 4 phones, which supported Bluetooth 4.0 with A2DP.
These devices collected BLE adverts
using the OpenSensing~\cite{opensensing} app for Android 4.2.2.
\iftr In particular, the phones ran the data collector
in the funf-v3 repo.  The phones broadcasted unique device ids every
200ms, and processed all received adverts in batches.  Phones
periodically requested a new location from their location service and
recorded them.
\fi

A total of 850 phones reported a total of 40.8M
adverts received from other participating devices. 
Each entry includes the device ids of receiver and sender, and a
timestamp. Each device recorded on average 48k
adverts (min: 7, median: 38.4k
max: 247.5k).
Active devices wake up every 5 minutes and process any received BLE
adverts, time\-stamping them at the moment of processing.
Consequently, the distribution of time intervals between recorded
batches for each device is very tightly clustered around 5min, with
the 75th and 95th percentile at 5min and 20min, respectively.

\section{Evaluation}
\label{sec:eval}

In this section, we present results of our experimental evaluation. We
ran {\sys} on an Intel(R) Xeon(R) CPU E7-8857 v2 @ 3.00GHz
  machine (4 sockets, 12 cores/socket, 1 thread/core) with 1.48TB RAM,
  running Debian 11.  TrustRank and the location proof algorithm are
implemented in Python. Location proofs use Open\-Street\-Map graphs
via the OSMnx 1.1.2 package \cite{osmnx}.  Next, we show how location
proof precision is affected by the encounter density and $N$. In
\cref{sec:evaltrustrank}, we show the effectiveness of \sys's defense
against retroactive collusion attacks.

\subsection{Location Proof Precision}
\label{sec:eval-feasible-regions}
Our objective is to understand how the location proof precision varies
with the number of required verifiers $N$ (higher $N$ should reduce
precision), the average location reporting frequency of devices
(higher frequency should improve precision) and the temporal encounter
density at the location of the proof (higher encounter density should
improve precision). A fourth factor, the {\sys} adoption rate, affects
precision indirectly by shifting the distribution of temporal
encounter densities: a higher adoption rate shifts the distribution
towards higher encounter densities, and, hence, should result in
improved precision on average.

We report on experiments with our simulated dataset, which is
representative of individuals' mobility and encounters in an actual
city. In all experiments, we ask for proofs that a device was within a
circular region with a given center (which coincides with the location
the device was visiting per our simulation) and a given radius $R$,
which represents the precision. For each proof, encounter selection is
set to find all encounters in a window of 10 minutes ($a = 5$ in the
encounter selection algorithm). Our graphs show the minimum radius $R$
(the highest precision) for which we could find a proof.

\btitle{Effect of $N$.}  \cref{fig:simulation-results} shows the
10th-, 50-th, and 90th-percentiles of the attainable precision radius
$R$ as a function of $N$ in several different experimental settings
(varying the average location reporting frequency, the encounter
density and the adoption rate). The percentiles are estimated by
running the algorithm of \cref{sec:service} on 5,000 randomly sampled
visits in the respective experimental settings.

In all experimental settings, $R$ increases with $N$, showing (as
expected) that the precision reduces as the required number of
independent verifiers increases.

\btitle{Effect of the average location reporting frequency.}
Next, we explore the effect of the average location reporting
frequency on the relation between $R$ and $N$. Each graph of
~\cref{fig:simulation-results-loc-freq} shows the 10th-, 50th- and
90th percentiles of $R$ as a function of $N$ for a different average
location reporting frequency -- 1 min, 3 min or 5 min. The adoption
rate is fixed at 20\% across the three graphs, and samples are drawn
from visits with encounter densities close to the median density for
this adoption rate (11 encounters in 10 minutes).

The distribution of precision radii (for any given $N$) shifts to the
higher side as the frequency of location reports increases from 1 to 5
minutes. This is also expected because, when peer location reports
bracketing an encounter are far apart in time, the location proof
algorithm of \cref{sec:service} results in larger isochrones.

As an example, for an average location reporting frequency of 3 min,
which is a reasonable trade-off between conserving device battery and
proof precision, the median of the attainable precision varies from $R
= 75$\,m for $N = 1$ to $R = 500$\,m for $N=5$. This shows that proofs
of reasonable precision are obtainable even with a 20\% adoption and
location reports as far apart as 3 min on average.

\btitle{Effect of encounter density.}
Next, we consider the effect of encounter density. For this
experiment, we fix the adoption rate and average location reporting
frequency at 20\% and 3 min, respectively. The three graphs of
~\cref{fig:simulation-results-encounter-density} show the percentiles
of $R$ for samples drawn from visits with different encounter
densities: 20, 40 and 50 encounters in 10 mins.
As expected, for any given $N$ precision radii are smaller (better
precision) when the encounter density is higher.

\btitle{Effect of adoption rate.}
Finally, we explore the effect of varying adoption rate on the
relation between $R$ and $N$. The three graphs of
~\cref{fig:simulation-results-adoption-rates} show the relation
between $R$ and $N$ for three different {\sys} adoption rates -- 10\%,
20\% and 40\%. The average location reporting frequency is fixed at 3
min, and the encounter density is set to the median encounter density
for the respective adoption rate. We observe that for a fixed $N$, $R$
reduces as the adoption rate increases. This is in line with what we
expect -- as the adoption rate increases from 10\% to 40\%, the median
encounter densities increase from 9 encounters in 10 minutes to 14
encounters in 10 minutes and the distribution of $R$ shifts downwards
(higher precision) accordingly.

\btitle{Summary.}  From these experiments, we conclude that
constructing location proofs with a precision of hundreds or even tens
of meters is possible for reasonable values of $N$ (upto 10),
reasonable location reporting frequencies (3 minutes on average), and
moderate adoption rates (20\%).

\begin{figure*}[h]
  \newcounter{figbackup}
  \setcounter{figbackup}{\value{figure}}
  \caption{Precision radius ($R$) as a function of the
    number of independent witnesses ($N$) for different
    location reporting frequencies, encounter densities
    and adoption rates.
    [Legend: \textcolor{blue}{10-th percentile}, \textcolor{red}{median}, \textcolor{green}{90-th percentile}].}
  \label{fig:simulation-results}
  \setcounter{figure}{\value{figbackup}}
  
  \begin{subfigure}{\textwidth}
    \caption{Varying average location reporting frequencies: 1, 3 and
      5 mins. Adoption rate fixed at 20\%. Points in each graph
      sampled from visits with encounter densities close to the median
      for this adoption rate, which is 11 encounters in 10 minutes.}
    \label{fig:simulation-results-loc-freq}
    \begin{tabular}{@{}ccc@{}}
      \includegraphics[width=0.32\textwidth]{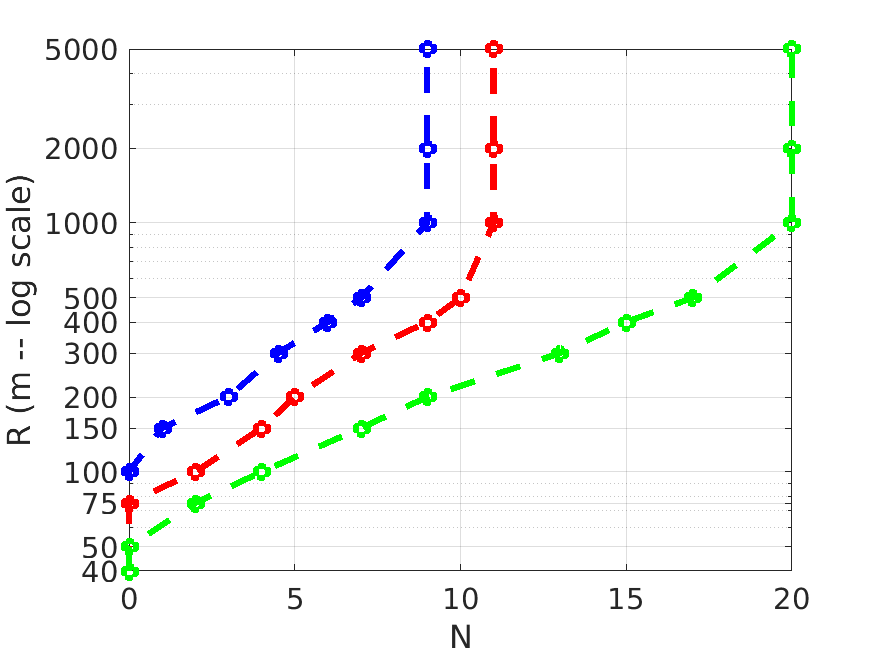}
      &
      \includegraphics[width=0.32\textwidth]{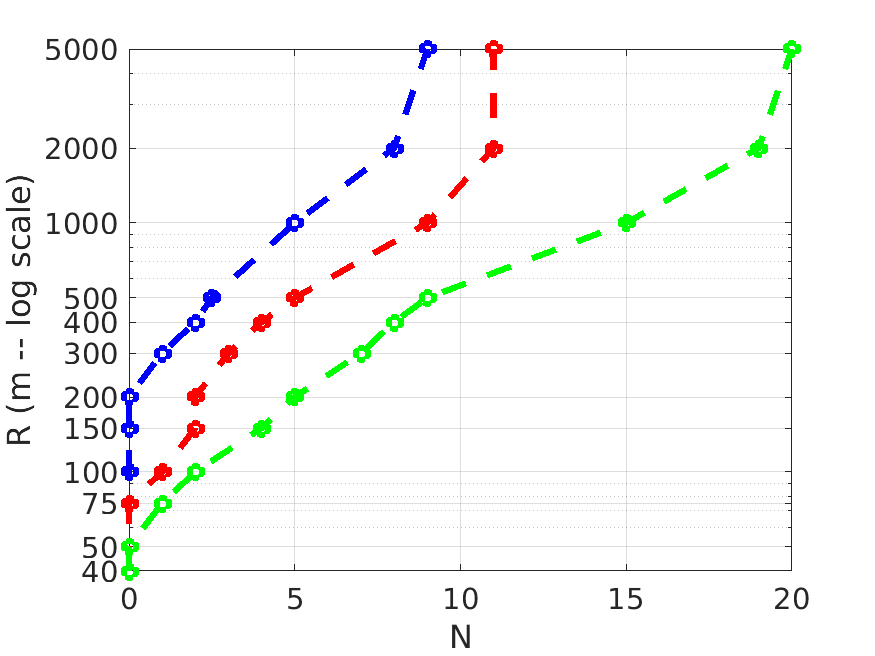}
      &
      \includegraphics[width=0.32\textwidth]{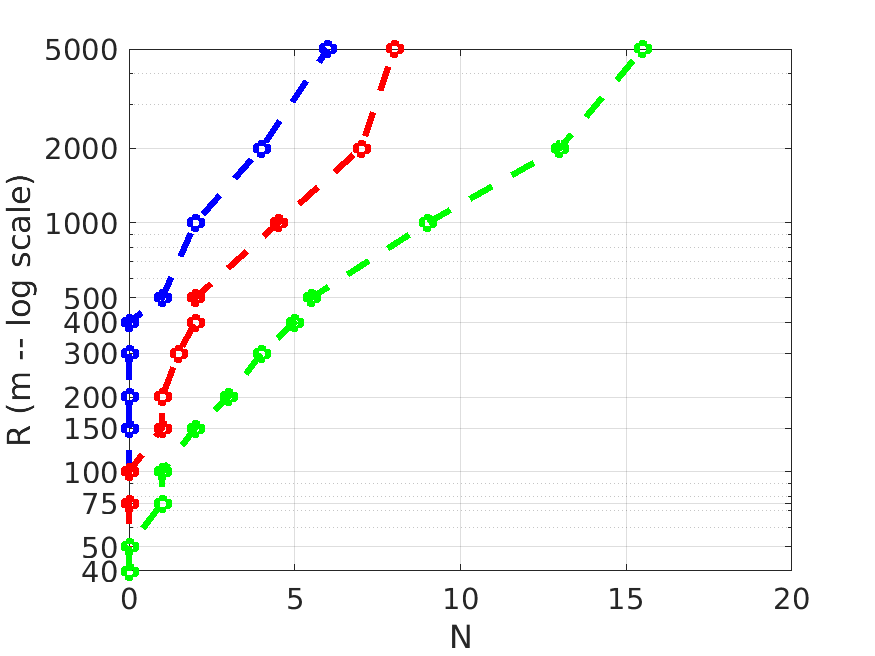}
    \end{tabular}
  \end{subfigure}

  \begin{subfigure}{\textwidth}
    \caption{Varying encounter densities: 20, 40 and 50
      encounters in 10 minutes. Sampled from a simulation with
      adoption rate 20\% and average location reporting frequency
      3 minutes.}
    \label{fig:simulation-results-encounter-density}
    \begin{tabular}{@{}ccc@{}}
      \includegraphics[width=0.32\textwidth]{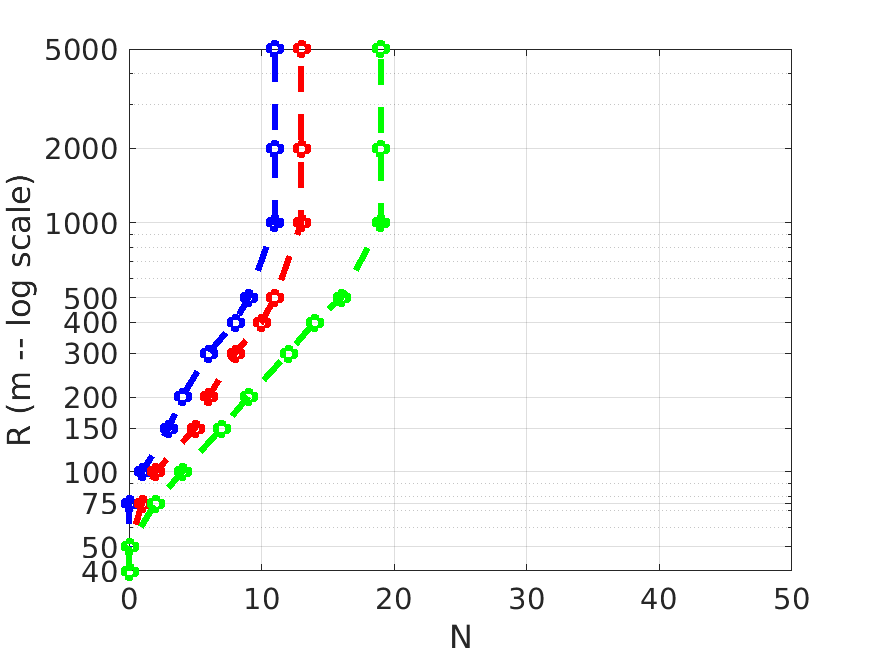}
      &
      \includegraphics[width=0.32\textwidth]{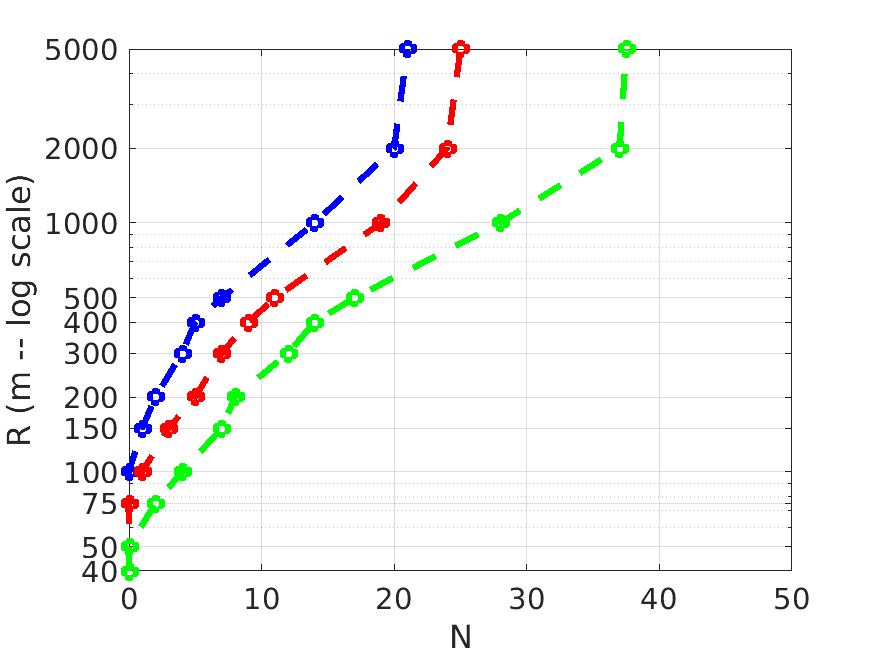}
      &
      \includegraphics[width=0.32\textwidth]{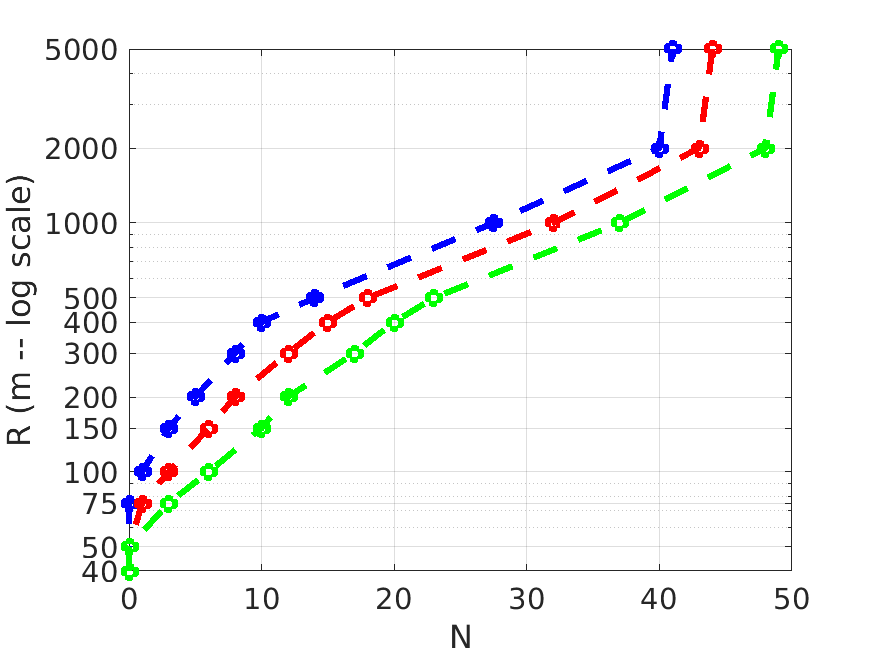}
    \end{tabular}
  \end{subfigure}

  \begin{subfigure}{\textwidth}
    \caption{Varying adoption rates: 10\%, 20\% and 40\%. Average
      location reporting frequency fixed at 3 minutes. Points in each
      graph sampled from visits with encounter densities close to the
      median for this location reporting frequency and the respective
      adoption rate (9, 11 and 14 encounters in 10 minutes).}
    \label{fig:simulation-results-adoption-rates}
    \begin{tabular}{@{}ccc@{}}
      \includegraphics[width=0.32\textwidth]{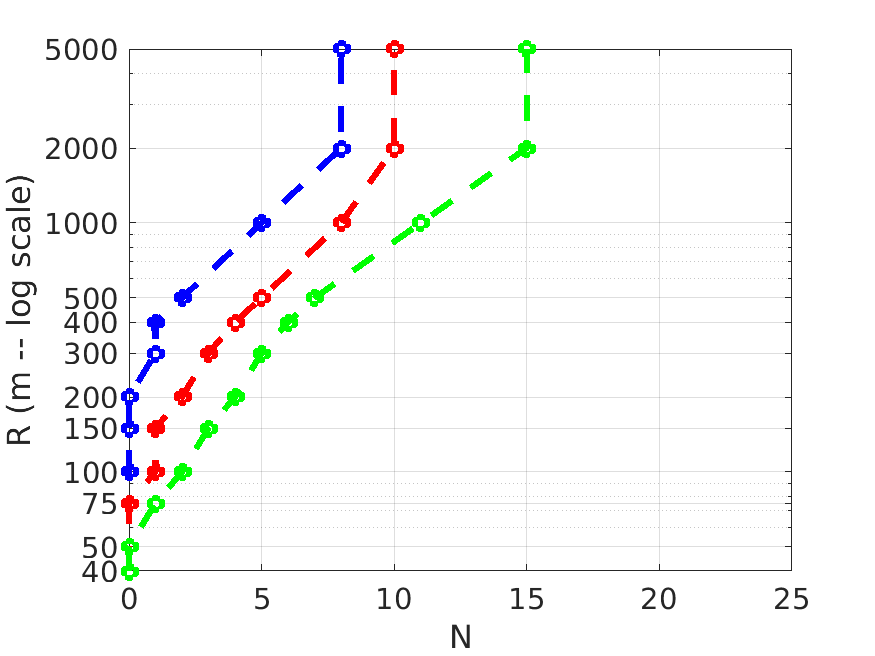}
      &
      \includegraphics[width=0.32\textwidth]{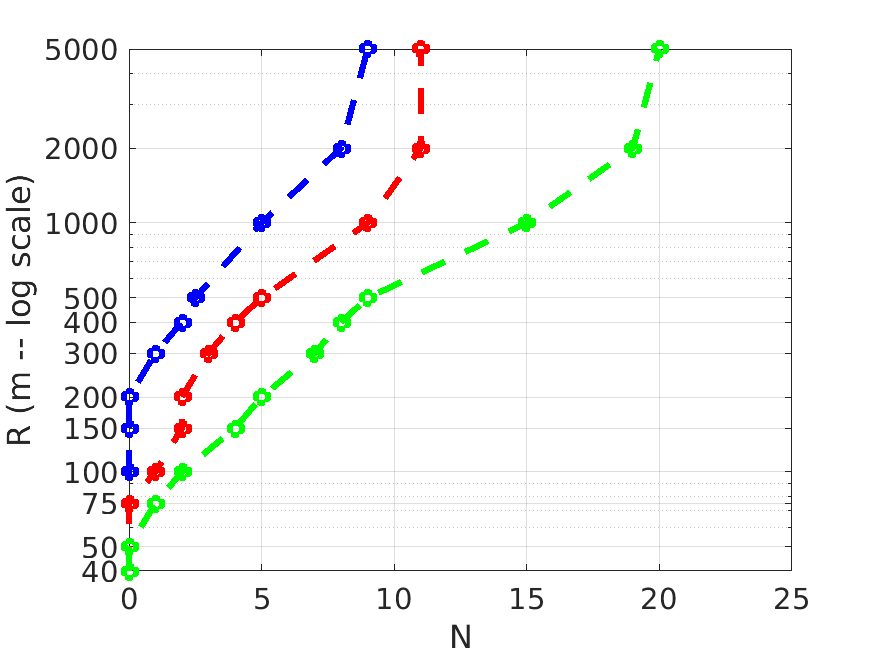}
      &
      \includegraphics[width=0.32\textwidth]{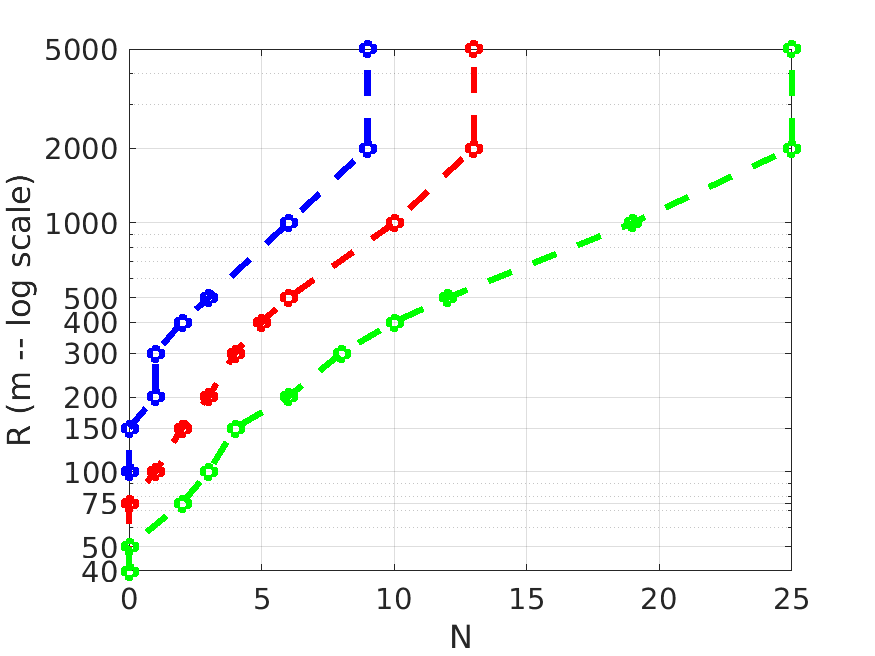}
    \end{tabular}
  \end{subfigure}
\end{figure*}

\subsection{TrustRank}
\label{sec:evaltrustrank}
Next, we evaluate the effectiveness of TrustRank in identifying
  fictitious and Sybil devices. Here, we use only the real-world
  dataset because our TrustRank algorithm relies on individual
  adverts, which are not available in our simulated dataset. Our
  real-world dataset includes only honest devices so, for this
  evaluation, we simulated several (in some ways optimally) powerful
  attacks on top of the real-world dataset.

\btitle{Simulated attacks on the real-world dataset.}  In each attack,
we first mark a number $c$ = 1, 2, 4 or 8 of randomly chosen devices
from among the 850 devices as \emph{corrupt}.  Since they are randomly
chosen from honest devices, corrupt devices' behavior is fundamentally
indistinguishable from honest behavior.

We then simulate $m$ Sybil devices on each corrupt device,
  varying $m$ from 1 to 128. We replicate every advert in the dataset
  whose source is a corrupt device to every Sybil on that device, thus
  simulating the Sybil attack of \cref{sec:attacks-sys},
  which gives the Sybils as many real encounters as possible.

Finally, for every device in our original dataset, we add an
adversary-created \emph{fictitious} device, which acts as its
doppelg\"anger.  The doppelg\"angers mimic the pairwise encounters and
reported locations of the real devices.  This is a very powerful
attacker that controls as many devices and uploads as much data as the
real world.  What's more, the behaviour of the fictitious devices is
indistinguishable from that of honest devices.  Without some trusted
devices, there is no way to know which encounter subgraph is honest,
and which is fictitious: they're symmetric.

Whenever a corrupt device receives an advert from a real device, it
also receives an advert from the corresponding doppelg\"anger. Trust
flows from the real world to the fictitious world along these bridging
adverts through corrupt devices.

Overall, we simulate strong attacks where the attacker corrupts up to 8 real devices (1\% of all devices), uses Sybils and uploads a fictitious world as large as the real world.

\btitle{TrustRank.}
We randomly mark 10 devices as trusted, run Trust\-Rank on the
encounter graphs for each of the above scenarios (we empirically picked the exponent
$L = 3$ for our edge weighting), and report our
observations. (Trust\-Rank is actually very robust to the
  initial choice of trusted devices so verifiers have a wide choice in
  the number and set of trusted devices they pick for location
  proofs. We experimented with different sets of trusted devices,
  ranging in size from 1 to 16, and the results were very similar in
  all cases.)

\begin{figure}[t]
	\centering
  \vspace{-3mm} %
  \caption{
    CDF of TrustRank values for honest and attacker devices
     (corrupt and fictitious).
    Our threshold (the vertical red line) is computed using data from
     many simulated attacks. %
  }
  \label{fig:analysis-cdf}
  \begin{tabular}{rrc}
    \rotatebox[origin=l]{90}{\hspace{0mm}fraction of devices with}\hspace{-1mm} &
    \hspace{-2mm}\rotatebox[origin=l]{90}{\hspace{1mm}at least X TrustRank}\hspace{-1mm} &
    \hspace{-2mm}\includegraphics[width=75mm]{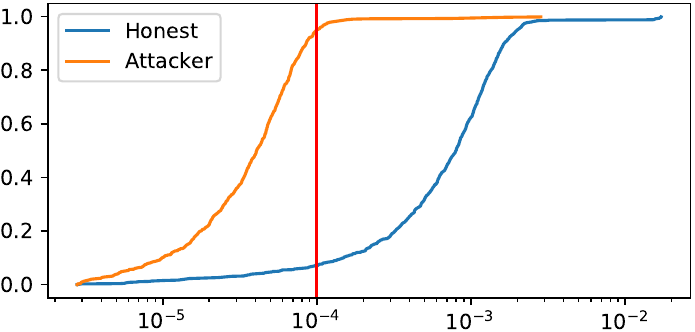}
    \vspace{-0.5mm}
    \\
    &
    &
    TrustRank
  \end{tabular}
  \vspace{-8mm}
\end{figure}

\btitle{TrustRank threshold.} We first examined the
distributions of TrustRank scores of honest and attacker (corrupt,
Sybil and fictitious) devices in different attack scenarios. As an
example, \cref{fig:analysis-cdf} shows the CDFs of the TrustRank
scores of honest and attacker devices with $c = 8$, $m = 1$. There is
a strong separation of over an order of magnitude between the
TrustRank scores of honest and attacker devices.

Next, we determine a \emph{uniform} cut-off threshold across all our
attack scenarios (the vertical red line in \cref{fig:analysis-cdf}), by
maximizing the sum of the true positive and the true negative rates
across all attack scenarios (see \cref{sec:trustrank}). For
our dataset, this optimal threshold was 0.0000995. As shown later,
this single threshold was very effective across \emph{all} our attack
scenarios. This indicates that, in practice, where the actual attack
is unknown, one could still determine a useful threshold by simulating
a limited number of possible attacks.

Above, we used ground truth about \emph{all} honest devices to
calculate the global threshold. In practice, the threshold can be
approximated by adding simulated attackers to a subset of the
  encounter data limited to a \emph{few} that are known to be honest.
To confirm that this suffices, we recomputed the threshold
on sample sub-datasets featuring only a small, 
randomly chosen subset of 25 honest 
devices (<3\% of devices) and all the attack devices. This
threshold was nearly as effective as our optimal threshold in
separating honest and adversary devices. 

\begin{table}
  \centering
\caption{TrustRank effectiveness with $c$ corrupt devices, each with $m$ Sybils.
  TP = fraction of correctly identified honest devices. TN
  = fraction of correctly identified virtual devices.}
\label{tab:trustrank-effectiveness}
\vspace{-4mm}
\begin{tabular}{|c|c|c|c|c|}
    \hline
    \multirow{2}{*}{c} & \multirow{2}{*}{m} & \multirow{2}{*}{TP} & \multicolumn{2}{c|}{TN}  \\
    \cline{4-5}
    &   &    & Sybil & Fictitious \\
    \hline
    \multirow{3}{*}{1}&1&0.929&0&1 \\
    &8&0.930&1&1 \\
    &16&0.930&1&1 \\
    \hline
    \multirow{3}{*}{2}&1&0.930&0.5&1 \\
    &8&0.931&1&1 \\
    &16&0.931&1&1 \\
    \hline
    \multirow{3}{*}{4}&1&0.929&0.25&0.998 \\
    &8&0.931&1&1 \\
    &16&0.931&1&1 \\
    \hline
    \multirow{3}{*}{8}&1&0.927&0.125&0.957 \\
    &8&0.930&1&1 \\
    &16&0.930&1&1 \\
    \hline
\end{tabular}
\vspace{-6mm}
\end{table}

\btitle{TrustRank effectiveness.}  \Cref{tab:trustrank-effectiveness}
shows that with the threshold computed above, we can correctly
identify most honest devices (column TP) and most virtual devices
(columns TN), across all our attack configurations (we do not
  show results for $m>16$ as those are nearly identical to those for
  $m=16$). The fraction of correctly identified honest devices is
consistently around 0.93, independent of the attack, again suggesting
that our selection of a global threshold is robust.

The fraction of correctly identified fictitious devices is
consistently above 0.998, even when $c = 4$ devices or nearly 0.5\% of
all real devices are corrupt. Even with $c=8$ (nearly 1\%
  corrupt devices), 0.957 of all fictitious devices can be detected.
Thus, for adversaries that are limited to a small fraction of corrupt
devices, which we assume to be the case in our defense, TrustRank is
nearly perfect at eliminating the threat of fictitious devices.

Finally, \cref{tab:trustrank-effectiveness} shows that our
edge-weighting scheme with $L = 3$ correctly discourages high numbers
of Sybils per host by harming the adversary: With 8 or more Sybils per
host, the TrustRank scores of all Sybils fall considerably, causing
all virtual devices to be detected.
This shows that our edge-weighting scheme works as intended.

As mentioned in \cref{sec:trustrank}, the ideal strategy of an
adversary is actually to activate the Sybils on a corrupt device one
at a time at epoch granularity, distributing trust as evenly as
possible among them. Assuming that the average TrustRank of a real
device is $s$ and the TrustRank threshold is $T$, the number of Sybils
an adversary can sustain per corrupt device without detection is
$s/T$. In our experiments, this ratio is close to 10.6 in all attack
scenarios, which indicates that it is not very sensitive to specific
attacks.

\subsection{Computational resources}
We ran TrustRank over the DTU dataset plus the simulated
attacks
in under 1min %
on a single, standard server-class machine.
As to the location proofs, their cost is dominated by the isochrone computation (about $N
\times 23s$). However, the current computation is not optimized: it computes each isochrone from scratch relying on off-the-shelf Python APIs, and it does not store results that can be (partially) re-used to compute future location proofs
in the same space-time region.

We believe that the computation can be considerably optimized
  for a production system, and scaled out easily. In fact, algorithms
  like TrustRank or PageRank are well-studied, and can run on enormous
  graphs (e.g., in production, TrustRank weights would be added to
each edge incrementally, as each new data point arrives).
Furthermore, the feasible region computation can be truncated
  once the feasible region grows beyond a pre-determined size (very
  large regions are not useful in practice), and the feasible region
  computation can also be trivially parallelized.

\section{Related Work}
\label{sec:related}

{\sys} differs from prior work on location proofs in two
ways. First, all prior work provides proofs of presence at a point
using interaction with fixed infrastructure or other devices
\emph{exactly} at the point of interest, which may not always
exist. {\sys} instead provides proofs of presence in a
(quantitatively sized) region using information about interaction with
other devices close to the region of interest, not at a specific
point, which is a more realistic requirement. Second, {\sys} provides
a defense against (retroactive) collusion attacks on location proofs
even when the adversary simulates virtual devices, a scenario that
prior work did not consider.

\btitle{Infrastructure-based location proofs.}
Saroiu et al.~\cite{saroiu2009enabling} articulate the importance of
location proofs as enablers for reliable location-based
services, and describe six applications. They sketch a method
to produce location proofs using WiFi access points (AP).
VeriPlace ~\cite{veriplace} also uses WiFi APs, but improves privacy by
hiding clients' identities from APs and allows users to disclose
approximate locations to services.
Hasan et al.~\cite{Hasan2011WhereHY} augment WiFi AP-based location
proofs with endorsements from nearby devices in BT range, to
raise the bar for collusion attacks among devices and APs. A similar
technique is described by Ferreira et al.~\cite{8548244}.

CREPUSCOLO~\cite{crepuscolo} relies on BT witnesses and trust\-ed
devices (TDs) placed at strategic locations that users are expected to
encounter frequently. If a location proof has a TD as witness, the
resulting proof is collusion-resistant.
However, proofs
for locations out of range of a TD rely on ordinary witnesses, and
offer no protection from collusion.

Many prior works have addressed infrastructure-based location proofs
in the context of vehicular
networks~\cite{6810011,9091099,Graham2009,9500728}.
Pham et al.~\cite{Pham14ubicomp} and Maia et al.~\cite{cross-city}
describe location proofs for specific applications, namely proof of
activity (e.g., distance/elevation traveled) and smart tourism (e.g.,
verifying that a tourist has visited certain sites), respectively.

 In contrast, {\sys} does not require infrastructure for location
  proofs; it relies solely on peer devices' histories.

\btitle{Infrastructure-less location proofs.}
Other prior work~\cite{link, stamp, applaus, globalattestation,pasport,props}
relies on BT witnesses and suggest or include collusion defenses.
APPLAUS~\cite{applaus} suggests weighted witness quorums where a device's trust level
depends on its history of witness choice, relative to the reported
density of location proofs nearby in space and time.
LINK's~\cite{link} per-device trust scores increase additively when a
device votes with the majority in a successful proof, and diminish
multiplicatively when the device acts suspiciously.
STAMP~\cite{stamp}
computes a per-device {\em entropy} value, which reflects the
diversity and randomness in the device's witness selection and is
additionally boosted whenever the device participates in a location
proof that is confirmed by a trusted device.
Arunkumar et al.~\cite{globalattestation}
check if the near-range radio contacts and trajectories reported by
all devices are consistent, and compute a trust score based on the
ratio of encounters in which the device reports consistent information.
A device's location is considered confirmed if the witness with the
highest trust score agrees. PROPS~\cite{props} and
PASPORT~\cite{pasport} do not consider collusion between witnesses and
provers; their defenses work when either the witnesses or the prover
is malicious.  All these systems are vulnerable to attacks with
fictitious devices, which can fabricate plausible trajectories,
encounters, witness selections, and even proofs involving trusted
devices using a small number of corrupt devices and Sybils, thus
acquiring trust to the point where they are indistinguishable from
honest devices.

\btitle{Hardware-based location proofs}
Some prior work relies on a hardware root-of-trust (e.g., TPM modules)
on mobile devices to ensure the integrity of location
readings~\cite{gilbert-hotm10,saroiu-hotm10,tpmlsb}.
Some work has addressed GPS spoofing in specific contexts like
aviation (Crowd-GPS-Sec~\cite{Jansen-SP18}) and mobile
games~\cite{Wong-ISA20}.
Gonz{\'a}lex-Tablas et al.~\cite{LAP-survey-05} provide a survey of
location authentication protocols and spatial-temporal attestation
services.  In contrast to these, \sys\ does not require trusted
  hardware on the devices.

\section{Conclusion}
\label{sec:conclusion}

ProLoc enables a mobile service to compute retroactive location proofs 
by integrating devices' trajectories and encounters.
ProLoc computes a feasible region, including the
locations at which the device could plausibly have resided at the
time, and it does so in a principled manner from the set of
encounters the device had around the time and its peers'
trajectories.
Moreover, {\sys} mitigates
large-scale collusion involving fictitious devices and Sybils,
effectively limiting the adversary's power to a small multiple of the
corrupt physical devices it controls. As a result, even a powerful
attacker is very limited in its ability to retroactively generate
false proofs. Finally, we sketch how to extend {\sys}'s defenses to
premeditated attacks, although a full treatment of such defenses
remains as future work.

\iftr
\appendix
\section{Edge-weighting scheme}
\label{sec:edge-weighting}

\begin{figure*}
\[
  e\left(d,g\right)\triangleq
  \min\left(\textrm{weight\_cap},
  \bigint_{\textrm{time } t}
  \frac{
    \left(\left\lbrace
    \begin{array}{ll}
      1 & \mbox{if $g$ recd.\ advert from $d$ in $[t-w/2, t+w/2]$} \\
      0 & \mbox{o/w}
    \end{array}
    \right.\right)
    \textrm{d}t
  }{
    \left(
    \begin{array}{c}
      \mbox{\# devices from which $g$}\\
      \mbox{ recd.\ adverts in $[t-w/2, t+w/2]$}
    \end{array}
    \right)^L
  }
  \right)
  \]
  \caption{Edge-weighting scheme of \cref{sec:trustrank} in continuous time}
  \label{fig:edge-weighting-continuous}
\end{figure*}

In \cref{sec:trustrank}, we described an approximation of our actual
edge-weighting scheme for defeating Sybil multiplicity attacks. Here,
we describe the actual scheme. This scheme also uses a time window $w$
(typically 8 minutes), but it does not divide time at each node into
discrete epochs to distribute edge weights. Instead, it considers
infinitesimally small epochs over intervals $[t, t+\textrm{d}t]$ at
each node $g$. In each such epoch, $g$ gives some weight to the
outgoing edge towards every device $d$ from which $g$ has received an
advertisement in the vicinity of $t$, specifically, in the interval
$[t-w/2, t+w/2]$. The weight given to each such edge in the epoch is
$dt/(\mbox{\# such devices})^L$. The total weight of the outgoing edge
towards $d$ is the \emph{integral} over time of these infinitesimally
small weights.

One additional difference from what we described in
\cref{sec:trustrank} is that we cap the weights of edges to prevent
any device from getting too much trust from a single device.
\Cref{fig:edge-weighting-continuous} summarizes our revised
edge-weight function, $e(d,g)$.

\fi
\bibliographystyle{plain}
\bibliography{main}

\begin{thebibliography}{10}

\bibitem{osmnx}
{Boeing, G. 2017. OSMnx: New Methods for Acquiring, Constructing, Analyzing,
  and Visualizing Complex Street Networks. Computers, Environment and Urban
  Systems 65, 126-139. doi:10.1016/j.compenvurbsys.2017.05.004 }.
\newblock \url{https://github.com/AMDESE/AMDSEV/tree/sev-snp-devel}.
\newblock Accessed: 2021-12-20.

\bibitem{britannica}
{Britannica. Citizen journalism.}
\newblock \url{https://www.britannica.com/topic/citizen-journalism}.
\newblock Accessed: 2023-11-15.

\bibitem{proloc-tr}
{ProLoc: Technical Report}.
\newblock \url{https://tinyurl.com/proloctr}.
\newblock Accessed: 2022-03-25.

\bibitem{yale}
{YaleGlobal Online. Citizens Fight Back with Cellphones and Blogs}.
\newblock
  \url{https://archive-yaleglobal.yale.edu/content/citizens-fight-back-cellphones-and-blogs}.
\newblock Accessed: 2023-11-15.

\bibitem{dtudatasets}
{Sensible DTU Project}.
\newblock \url{https://github.com/SocialComplexityLab/sensible_docs}, 2018.
\newblock Accessed on 28/12/2020.

\bibitem{opensensing}
{OpenSensing GitHub}.
\newblock \url{https://github.com/OpenSensing}, 2020.
\newblock Accessed on 13/01/2020.

\bibitem{openstreetmap}
{OpenStreetMap}.
\newblock \url{https://www.openstreetmap.org}, 2021.
\newblock Accessed on 19/08/2021.

\bibitem{dataforgoodmeta}
{Data for Good at Meta}.
\newblock \url{https://data.humdata.org/organization/meta?}, 2023.
\newblock Accessed on 27/11/2023.

\bibitem{DBLP:conf/sp/AlvisiCELP13}
Lorenzo Alvisi, Allen Clement, Alessandro Epasto, Silvio Lattanzi, and
  Alessandro Panconesi.
\newblock {SoK: The Evolution of Sybil Defense via Social Networks}.
\newblock In {\em 2013 {IEEE} Symposium on Security and Privacy, {SP} 2013,
  Berkeley, CA, USA, May 19-22, 2013}, pages 382--396. {IEEE} Computer Society,
  2013.

\bibitem{globalattestation}
S.~{Arunkumar}, M.~{Srivatsa}, M.~{Sensoy}, and M.~{Rajarajan}.
\newblock Global attestation of location in mobile devices.
\newblock In {\em MILCOM 2015 - 2015 IEEE Military Communications Conference},
  pages 1612--1617, 2015.

\bibitem{pancastnature}
Gilles Barthe, Roberta~De Viti, Peter Druschel, Deepak Garg, Manuel
  Gomez-Rodriguez, Pierfrancesco Ingo, Heiner Kremer, Matthew Lentz, Lars
  Lorch, Aastha Mehta, and Bernhard Sch\"{o}lkopf.
\newblock Listening to bluetooth beacons for epidemic risk mitigation.
\newblock {\em Scientific Reports}, 12, 2022.

\bibitem{9091099}
Mohamed Baza, Mahmoud Nabil, Mohamed Mohamed Elsalih~Abdelsalam Mahmoud, Niclas
  Bewermeier, Kemal Fidan, Waleed Alasmary, and Mohamed Abdallah.
\newblock {Detecting Sybil Attacks using Proofs of Work and Location in
  VANETs}.
\newblock {\em IEEE Transactions on Dependable and Secure Computing}, pages
  1--1, 2020.

\bibitem{crepuscolo}
Eyüp~S. Canlar, Mauro Conti, Bruno Crispo, and Roberto Di~Pietro.
\newblock {CREPUSCOLO: A collusion resistant privacy preserving location
  verification system}.
\newblock In {\em 2013 International Conference on Risks and Security of
  Internet and Systems (CRiSIS)}, pages 1--9, 2013.

\bibitem{8548244}
João Ferreira and Miguel~L. Pardal.
\newblock {Witness-Based Location Proofs for Mobile Devices}.
\newblock In {\em 2018 IEEE 17th International Symposium on Network Computing
  and Applications (NCA)}, pages 1--4, 2018.

\bibitem{props}
Sébastien Gambs, Marc-Olivier Killijian, Matthieu Roy, and Moussa Traoré.
\newblock {PROPS : A PRivacy-preserving lOcation Proof System}.
\newblock In {\em Proceedings of the IEEE 33rd International Symposium on
  Reliable Distributed Systems}, pages 293--307, 2014.

\bibitem{gilbert-hotm10}
Peter Gilbert, Landon~P. Cox, Jaeyeon Jung, and David Wetherall.
\newblock {Toward Trustworthy Mobile Sensing}.
\newblock In {\em Proceedings of the Eleventh Workshop on Mobile Computing
  Systems \& Applications}, HotMobile '10, page 31–36, New York, NY, USA,
  2010. Association for Computing Machinery.

\bibitem{LAP-survey-05}
A.~I. Gonz{\'a}lez-Tablas, K.~Kursawe, B.~Ramos, and A.~Ribagorda.
\newblock Survey on location authentication protocols and spatial-temporal
  attestation services.
\newblock In Tomoya Enokido, Lu~Yan, Bin Xiao, Daeyoung Kim, Yuanshun Dai, and
  Laurence~T. Yang, editors, {\em Embedded and Ubiquitous Computing -- EUC 2005
  Workshops}, pages 797--806, Berlin, Heidelberg, 2005. Springer Berlin
  Heidelberg.

\bibitem{Graham2009}
Michelle Graham and David Gray.
\newblock {{Protecting Privacy and Securing the Gathering of Location Proofs --
  The Secure Location Verification Proof Gathering Protocol}}.
\newblock In Andreas~U. Schmidt and Shiguo Lian, editors, {\em Security and
  Privacy in Mobile Information and Communication Systems}, pages 160--171,
  Berlin, Heidelberg, 2009. Springer Berlin Heidelberg.

\bibitem{trustrank}
Zoltan Gyongyi, Hector Garcia-Molina, and Jan Pedersen.
\newblock {Combating web spam with trustrank}.
\newblock In {\em Proceedings of the 30th international conference on very
  large data bases (VLDB)}, 2004.

\bibitem{Hasan2011WhereHY}
R.~Hasan and R.~Burns.
\newblock {Where Have You Been? {S}ecure Location Provenance for Mobile
  Devices}.
\newblock {\em ArXiv}, abs/1107.1821, 2011.

\bibitem{haynes2020gravity}
Kingsley~E Haynes and A~Stewart Fotheringham.
\newblock Gravity and spatial interaction models.
\newblock 2020.

\bibitem{Jansen-SP18}
K.~{Jansen}, M.~{Schäfer}, D.~{Moser}, V.~{Lenders}, C.~{Pöpper}, and
  J.~{Schmitt}.
\newblock {{Crowd-GPS-Sec}: Leveraging Crowdsourcing to Detect and Localize
  {GPS} Spoofing Attacks}.
\newblock In {\em 2018 IEEE Symposium on Security and Privacy (SP)}, pages
  1018--1031, 2018.

\bibitem{DBLP:conf/dsn/JiaWG17}
Jinyuan Jia, Binghui Wang, and Neil~Zhenqiang Gong.
\newblock {Random Walk Based Fake Account Detection in Online Social Networks}.
\newblock In {\em 47th Annual {IEEE/IFIP} International Conference on
  Dependable Systems and Networks, {DSN} 2017, Denver, CO, USA, June 26-29,
  2017}, pages 273--284. {IEEE} Computer Society, 2017.

\bibitem{9500728}
Qinglei Kong, Feng Yin, Yue Xiao, Beibei Li, Xuejia Yang, and Shuguang Cui.
\newblock {Achieving Blockchain-based Privacy-Preserving Location Proofs under
  Federated Learning}.
\newblock In {\em ICC 2021 - IEEE International Conference on Communications},
  pages 1--6, 2021.

\bibitem{SDDR}
Matthew Lentz, Viktor Erd{\'{e}}lyi, Paarijaat Aditya, Elaine Shi, Peter
  Druschel, and Bobby Bhattacharjee.
\newblock {{SDDR:} {L}ight-Weight, Secure Mobile Encounters}.
\newblock In {\em Proceedings of the 23rd {USENIX} Security Symposium}, pages
  925--940, 2014.

\bibitem{veriplace}
Wanying Luo and Urs Hengartner.
\newblock {VeriPlace}: A privacy-aware location proof architecture.
\newblock pages 23--32, 01 2010.

\bibitem{cross-city}
Gabriel~A. Maia, Rui~L. Claro, and Miguel~L. Pardal.
\newblock {{CROSS City: Wi-Fi Location Proofs for Smart Tourism}}.
\newblock In Luigi~Alfredo Grieco, Gennaro Boggia, Giuseppe Piro, Yaser
  Jararweh, and Claudia Campolo, editors, {\em Ad-Hoc, Mobile, and Wireless
  Networks}, pages 241--253, Cham, 2020. Springer International Publishing.

\bibitem{pasport}
Mohammad~Reza Nosouhi, Keshav Sood, Shui Yu, Marthie Grobler, and Jingwen
  Zhang.
\newblock {PASPORT: A Secure and Private Location Proof Generation and
  Verification Framework}.
\newblock In {\em IEEE Transactions on Computational Social Systems}, pages
  293--307, 2020.

\bibitem{tpmlsb}
H.~{Othman}, H.~{Hashim}, M.~A. {Yuslan Razmi}, and J.~{Ab Manan}.
\newblock {Forming Virtualized Secure Framework for Location Based Services
  ({LBS}) using Direct Anonymous Attestation ({DAA}) protocol}.
\newblock In {\em 2010 IEEE 6th International Conference on Wireless and Mobile
  Computing, Networking and Communications}, pages 622--629, 2010.

\bibitem{Pham14ubicomp}
Anh Pham, K\'{e}vin Huguenin, Igor Bilogrevic, and Jean-Pierre Hubaux.
\newblock {Secure and Private Proofs for Location-Based Activity Summaries in
  Urban Areas}.
\newblock In {\em Proceedings of the 2014 ACM International Joint Conference on
  Pervasive and Ubiquitous Computing}, UbiComp '14, page 751–762, New York,
  NY, USA, 2014. Association for Computing Machinery.

\bibitem{saroiu2009enabling}
Stefan Saroiu and Alec Wolman.
\newblock {Enabling new mobile applications with location proofs}.
\newblock In {\em Proceedings of the 10th workshop on Mobile Computing Systems
  and Applications}, pages 1--6, 2009.

\bibitem{saroiu-hotm10}
Stefan Saroiu and Alec Wolman.
\newblock {I Am a Sensor, and I Approve This Message}.
\newblock In {\em Proceedings of the Eleventh Workshop on Mobile Computing
  Systems \& Applications}, HotMobile '10, page 37–42, New York, NY, USA,
  2010. Association for Computing Machinery.

\bibitem{link}
Manoop Talasila, Reza Curtmola, and Cristian Borcea.
\newblock {{LINK: Location Verification through Immediate Neighbors
  Knowledge}}.
\newblock In Patrick S{\'e}nac, Max Ott, and Aruna Seneviratne, editors, {\em
  Mobile and Ubiquitous Systems: Computing, Networking, and Services}, pages
  210--223, Berlin, Heidelberg, 2012. Springer Berlin Heidelberg.

\bibitem{enclosure}
Lillian Tsai, Roberta De~Viti, Matthew Lentz, Stefan Saroiu, Bobby
  Bhattacharjee, and Peter Druschel.
\newblock {enClosure: Group Communication via Encounter Closures}.
\newblock In {\em Proceedings of the 17th Annual International Conference on
  Mobile Systems, Applications, and Services}, pages 353--365, 2019.

\bibitem{stamp}
X.~{Wang}, A.~{Pande}, J.~{Zhu}, and P.~{Mohapatra}.
\newblock {STAMP}: Enabling privacy-preserving location proofs for mobile
  users.
\newblock {\em IEEE/ACM Transactions on Networking}, 24(6):3276--3289, 2016.

\bibitem{Wong-ISA20}
Shing~Ki Wong and Siu~Ming Yiu.
\newblock {Detection on {GPS} Spoofing in Location Based Mobile Games}.
\newblock In Ilsun You, editor, {\em Information Security Applications}, pages
  215--226, Cham, 2020. Springer International Publishing.

\bibitem{6810011}
Yifan Zhang, Chiu~C. Tan, Fengyuan Xu, Hao Han, and Qun Li.
\newblock {VProof: Lightweight Privacy-Preserving Vehicle Location Proofs}.
\newblock {\em IEEE Transactions on Vehicular Technology}, 64(1):378--385,
  2015.

\bibitem{applaus}
Z.~{Zhu} and G.~{Cao}.
\newblock {APPLAUS}: A privacy-preserving location proof updating system for
  location-based services.
\newblock In {\em 2011 Proceedings IEEE INFOCOM}, pages 1889--1897, 2011.

\end{thebibliography}

\end{document}